\documentclass[11pt, a4paper]{amsart}
\usepackage{a4wide}
\usepackage{amssymb}
\usepackage{amsthm}
\usepackage{amsfonts}
\usepackage{color}
\usepackage{tikz}
\usetikzlibrary{arrows}
\newtheorem{lemma}{\bf Lemma}[section]
\newtheorem{theorem}[lemma]{\bf Theorem}
\newtheorem{corollary}[lemma]{\bf Corollary}
\newtheorem{proposition}[lemma]{\bf Proposition}

\newtheorem{definition}[lemma]{\bf Definition}
\newtheorem{remark}[lemma]{\bf Remark}
\theoremstyle{remark}
\newtheorem{example}[lemma]{\bf Example}
\newtheorem{algorithm}[lemma]{\bf Algorithm}

\newcommand{\In}{{\mathrm{In}}}
\newcommand{\Fill}{{\mathrm{Dtm}}}
\newcommand{\pFill}{{\mathrm{pDtm}}}

\begin{document}

\title{Graph-theoretic autofill}

\author{Michael Mayer}
\address{Consult AG, CH-8050 Zurich, Switzerland}
\email{\tt michael.mayer@consultag.ch}
\author{Dominic van der Zypen}
\address{Federal Office of Social Insurance, CH-3003 Bern,
Switzerland}
\email{\tt dominic.zypen@gmail.com}


\begin{abstract} 
Imagine a website that asks the user to fill in a web form and -- based on 
the input values -- derives a relevant figure, for instance an expected 
salary, a medical diagnosis, or the market value of a house. 
How to deal with missing input values at run-time? 
Besides using fixed defaults, a more sophisticated approach is
to use predefined dependencies (logical or correlational) between different fields to autofill 
missing values in an iterative way.
Directed loopless graphs (in which cycles are allowed) are the ideal mathematical 
model to formalize these dependencies.
We present two new graph-theoretic approaches to filling missing values at run-time.
\end{abstract}

\maketitle
\parindent = 0mm
\parskip = 2 mm

\section{Introduction}
The Internet offers many online calculators that provide 
relevant figures based on the values entered in a web form. 
Examples are salary calculators, web-based medical advice, and
tools to compute the typical rent of an appartment, just to name a few. 
Usually, all input fields are mandatory, i.\,e.\ cannot be left blank by the user. 
While this minimizes programming effort by the website provider, 
it might force the user to make up or guess unknown information (like the living area of a 4-room appartment 
whose monthly rent is to be estimated by the calculator) just for 
the sake of completeness, 
or the user will even search the web for a simpler calculator 
without ever returning. 

More user-friendly web forms offer fixed default values 
at least for some of the fields. On one hand, this approach is 
attractive thanks to its low programming effort. On the other hand, 
the stronger the input fields are interrelated, 
the less appropriate a fixed default might be in certain cases. 
While a default living area of $80\,m^2$ might serve as a 
good default for a typical appartment, it is 
certainly unrealistic for a studio 
or a 10-room penthouse with indoor swimming-pool. 

What are alternatives to handle missing input values on web forms? 
One option is the reduced feature approach from statistical 
predictive modelling, see e.\,g.\ \cite{SP} or \cite{SG}. There, 
for each combination of available fields, an individual calculator is applied.
The drawback is obvious. Even for a low number $p$ of optional input fields, 
the programming effort exploses, as $2^p$ calculators have to be derived, implemented and supported. 

An elegant compromise between offering over-simplistic fixed defaults and the 
unfeasible reduced feature approach is to autofill missing values by 
prespecified functions of other input values, i.\,e.\ using dynamic 
defaults for certain fields. This can either be made visible to the user or 
invisibly applied before calling the underlying formula of the calculator. 

Consider as example a calculator for the ideal weight based on body height, sex and age. There, the autofill strategy could be as follows.
\begin{itemize}
	\item Sex $s$: Mandatory (1: male, 0: female)
	\item Age $x$: Autofilled by height $z$ in cm using the formula
	\begin{align*}
	   x = f(z) := 
	   \begin{cases} 
	   	 \lfloor (z-30)/130 \cdot 16 + 1 \rfloor & \text{if } 30 \le z \le 160, \\
	   	 40 & \text{if } z > 160, \\
	   	  1 &  \text{if } z < 30.
	   	\end{cases}
	\end{align*}
	\item Body height $z$: Autofilled by height and sex by
	\begin{align*}
	  z = g(x, s) := 
	  \begin{cases}
	    162 + 16s & \text{if } x > 16, \\
	    \lfloor (x - 1)/16 \cdot 130 + 30.5\rfloor & \text{if } x \le 16.
	  \end{cases}
	\end{align*}
\end{itemize} 
Thus, the web form with the fields ``age'', ``sex'' and ``height'' can be considered to be filled (in the sense as the underlying weight formula can be applied) as soon as sex and either of height and age is provided. Or to express the same thought in its negative sense: Even if we have specified replacement functions for age and height, the web form cannot be autofilled if age as well as height are left blank. In more complex situations, also the order in which the replacement functions are applied could matter, e.\,g.\ if sex would be autofilled by a function of the height.

The purpose of this article is to use the notion of directed graphs (in which cycles are allowed) to provide a theoretical framework to determine if an arbitrarily complex web form (or any other multivariate input, e.\,g.\ the argument list of a function written in a programming language like C) can be filled by a fixed set of replacement functions and the partial input provided so far by the user. To do so, each input field is associated with a vertex in a graph and each replacement function with at least one directed edge between these vertices. Note that our considerations do not depend on how the replacement functions are defined or how accurate they are. In practice they will be chosen based on expert knowledge, literature, logical rules or by exploring statistical relationships.
	
Suppose we are given an observation in which knowing the value
of field $A$ would enable you to guess the value of
field $B$. A natural way to represent
this is to draw an {\bf arrow} from $A$ to $B$:
$$A \to B$$
Notice that the relation ``$B$ can be guessed if we know $A$'' 
is sometimes asymmetric, that is, only works in one direction,
as the following example illustrates:

Let $A$ stand for ``gender'' and $B$ for ``number 
of pregnancies had so far''. Then if you know that some
observation represents a male then you can infer
that the value of $B$ must be $0$. On the other hand
if we know that the number of pregnancies is $0$, 
the individual at hand can be either male or female, 
depending on the dataset at hand even with roughly equal
probability.

So this shows that it is natural to choose {\em directed
graphs}, essentially points connected by
arrows, as a model for representing the conclusions
that can be made between some fields 
based on replacement functions. The terms of a directed
graph, and other terms, will be defined rigorously 
in the next section, and for a good introduction
into graph theory, we point the reader to \cite{Di}.

The left side of Figure~\ref{fig:intro} shows the graph corresponding to the example of the weight-calculator.
\begin{figure}
	\centering
	\resizebox{0.3\textwidth}{!}{%
	\begin{tikzpicture}[->,>=stealth',shorten >=1pt,auto,node distance=3cm, thick,main node/.style={circle,draw}]	
		\node[main node] (1) {Sex};
		\node[main node] (2) [below left of=1] {Age};
		\node[main node] (3) [below right of=2] {Height};
		
		\path[every node/.style={font=\sffamily\small}]
		(1) edge node [left] {} (3)
		(2) edge [bend right] node {} (3)
		(3) edge [bend right] node {} (2);
	\end{tikzpicture}}
	\hfill
	\resizebox{0.5\textwidth}{!}{%
		\begin{tikzpicture}[->,>=stealth',shorten >=1pt,auto,node distance=3cm, thick,main node/.style={circle,draw}]	
		\node[main node] (1) {Sex};
		\node[main node] (2) [below left of=1] {Age};
		\node[main node] (3) [below right of=2] {Height};
		\node[main node] (4) [below right of=1] {Pregnant};
		
		\path[every node/.style={font=\sffamily\small}]
		(1) edge node [left] {} (3)
		(3) edge [bend right] node {} (2)
		(4) edge [bend right] node {} (1)
		(4) edge [bend right] node {} (2)
		(4) edge [bend right] node {} (3)
		(1) edge [bend right] node {} (4)
		(2) edge [bend right] node {} (4);
		\end{tikzpicture}}
	\caption{On the left: Graph for weight-calculator. On the right: A more complex situation with additional field.}
\end{figure}
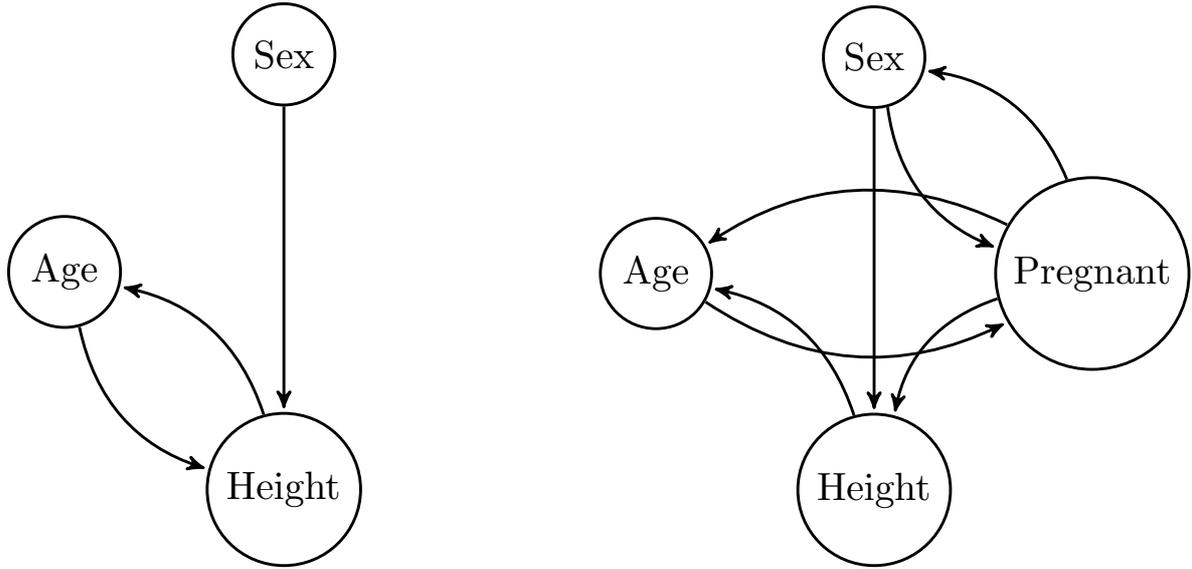\label{fig:intro}

We consider two kinds of replacement functions. 
In Section~\ref{se:theory:complete} we look at
the case of ``complete determination'', that is we are only allowed 
to calculate the value 
of a certain field $Y$ if {\em all} the values
of the fields from which an arrow leads into $Y$
are known or already derived from other fields. Thus, a replacement function can only be evaluated if all arguments are 
made available in some way as in the case of the weight-calculator. Next, Section~\ref{se:theory:partial} 
applies to the situation of ``partial determination'',
where we are allowed to calculate the value
of a certain field $Y$ if {\em only some} of the values
of the fields from which an arrow leads into $Y$
are known or derived from other fields, i.\,e.\ some missing or undetermined arguments in the replacement functions are allowed.
\section{Filling with complete determination}\label{se:theory:complete}

A {\em directed self-loopless graph} (which we will refer
to as {\em directed graph} for simplicity) is a tuple $G=(V,E)$ where $V$ is a
finite non-empty set and $$E\subseteq V\times V
\setminus\{(v,v): v\in V\}.$$ 
We call the set $V$ the set of {\em vertices}, and they represent
the data fields or variables. An {\em edge} represents
a determination arrow: If $(v,w)\in V$ we can make some 
conclusion from the variable $v$ to variable $w$.

\begin{definition}
Let $x,y\in V$. Then there is {\em a directed
path from $x$ to $y$} if there is $n\in\mathbb{N}$
and a map $p:\{0,\ldots,n\}\to V$ such that 
\begin{enumerate}
\item $p(0) = x, p(n) =y$;
\item for $k\in \{0,\ldots,n-1\}$ we have $\big(p(k), p(k+1)\big) \in E$.
\end{enumerate}
If $n\geq 1$ and $p:\{0,\ldots,n\}\to V$ is a directed path 
from $x$ to $x$
we call $p(\{0,\ldots,n\})\subseteq V$, 
that is the image of $p$, a {\em cycle}.
\end{definition}
Note that the above definition implies that a cycle has at least 
2 elements.
\begin{definition}\label{def_in}
For $v\in V$ we denote the set of {\em incoming vertices} by
$\In(v) = \{w\in V: (w,v)\in E\}$ and we let $A(G) = \{v\in V: 
\In(v) = \emptyset\}$. 
\end{definition}
$A(G)$ represents the fields for which no replacement 
function are defined. Thus, they are always mandatory (or, without 
affecting our theory, need a constant default).
\begin{definition}
Let $I\subseteq V$ and let $v\in V$. Then we say that 
$I$ {\em determines} the vertex $v$
if $\In(v)\neq \emptyset$ and $\In(v) \subseteq I$.
\end{definition}
In other words, $I$ stands for the arguments of the 
replacement function for $v$ which is represented in the graph by the determination arrows 
that come in from $\In(v)$.
\begin{definition}
We denote the set of vertices determined by $I$  by
$\Fill(I)$.
Inductively we define the ``determination closure'' 
of $I\subseteq V$ in the following way:
\begin{enumerate}
\item $I_{0} := I$;
\item for $n\geq 0$ set $I_{n+1} :=I_n \cup \Fill(I_{n})$.
\end{enumerate}
We say that $I\subseteq V$ {\bf fills} $V$ if there is $n\in
\mathbb{N}$ such that $I_{n} = V$. Note that trivially,
$V$ itself fills $V$.
\end{definition}
Consequently, if all fields associated with $I$ are entered, then the remaining fields can be autofilled 
iteratively by the prespecified replacement functions.

The following lemmata are mainly used to derive a characterization theorem for filling subsets.
\begin{lemma}\label{emptylem}
If $I\subseteq V$ fills $V$ then
$A(G) \subseteq I$.
\end{lemma}
\begin{proof}
Trivial, follows from definition of determination.
\end{proof}
\begin{lemma}\label{baslem}
Suppose $I\subseteq V$ and suppose that $C$ is a
cycle with $I \cap C = \emptyset$. Then
$$I_1 \cap C = \big(I \cup \Fill(I)\big) \cap C = \emptyset.$$
\end{lemma}
\begin{proof} We prove the contrapositive.
Suppose there is $c^*\in C$ such that $c^* \in I_1$. Since
$C$ is a cycle, we have $$C \cap \In(c^*) \neq \emptyset,$$
say $d \in C \cap \In(c^*)$.
Now $c^* \in I_1$ implies $\In(c^*) \subseteq I$ by 
definition, therefore $$d\in \In(c^*) \cap C \subseteq I\cap C,$$ 
so $I\cap C\neq \emptyset$.
\end{proof}
An inductive application of Lemma~\ref{baslem} shows that
any subset that fills $V$ intersects every cycle in $G=(V,E)$:
\begin{lemma}\label{cyclelem}
If $I\subseteq V$ fills $V$ and $C\subseteq V$
is a cycle, then $I\cap C \neq \emptyset$.
\end{lemma}
\begin{proof}
Let $I\subseteq V$ be any subset of $V$ and let $C\subseteq V$
be a cycle. Applying Lemma~\ref{baslem} inductively implies that
\begin{quote}
$(\star)$\hspace*{4mm} if $I\cap C = \emptyset$ then 
$I_n \cap C = \emptyset$
for all $n\in\mathbb{N}$.
\end{quote}
So if $I$ fills $V$
then $I_n=V$ for some $n$, and in particular
$I_n\cap C \neq \emptyset$ for that $n$.
The contrapositive of $(\star)$ directly implies
$I\cap C \neq \emptyset$.
\end{proof}
This helps us to prove a characterization theorem for 
filling subsets $I\subseteq V$.
\begin{theorem}\label{chara1}
Let $G=(V,E)$ be a self-loopless directed graph, and let $I\subseteq V$.
Then the following statements are equivalent:
\begin{enumerate}
\item $I$ is filling; 
\item $A(G)\subseteq I$, and for every cycle $C$ in $(V,E)$ 
we have $C\cap I \neq \emptyset$.
\end{enumerate}
\end{theorem}
A web form can thus be autofilled as soon as values are provided 
\begin{enumerate}
\item for all fields without replacement function, and 
\item for one field per cycle.
\end{enumerate}
\begin{proof}[Proof of Theorem~\ref{chara1}]
	
(1) $\implies$ (2).  Taken care of by Lemmas \ref{emptylem} and \ref{cyclelem}.

(2) $\implies$ (1). We assume that $I\subseteq V$ is not filling and $v\in I$ for
every $v\in V$ with $\In(v)=\emptyset$ and construct a cycle
$C^*$ that does not intersect $I$ (i.e. $C^*\cap I = \emptyset$).

Since $I$ is not filling, we have $I_n \neq V$ for all 
$n\in\mathbb{N}$. Since $V$ is finite, the increasing
sequence $$I = I_0\subseteq I_1\subseteq I_2 \subseteq \ldots$$
stabilizes at some $N\in\mathbb{N}$, that is there is
$N\in\mathbb{N}$ such that $I_k = I_N$ for all $k\geq N$, 
and of course $I_N\neq V$.

So we pick $z_0\in V\setminus I_N$.

Note that $\In(z_0) \neq \emptyset$ because all vertices
$v$ with $\In(v) =\emptyset$ are included in $I$ by assumption,
and $z_0\notin I$. Since $I_{N+1} =  I_N$ we have that $z_0$ is 
not determined by $I_N$ which means $\In(z_0)\not\subseteq I_N$. 
So we pick $z_1 \in \In(z_0)\setminus I_N$.

Inductively, and similarly to above, we pick $z_{k+1}\in 
\In(z_k)\setminus I_N$ for all $k\in\mathbb{N}$.

Then we consider the set $Z = \{z_k: k\in\mathbb{N}\}$.
By construction $Z$ has empty intersection with $I_N$ and
therefore also with $I\subseteq I_N$.  Moreover we have 
$(z_{k+1}, z_k)\in E$ for all $k$. Because $V$ is finite,
the set $Z$ must contain a cycle $C^*$ and 
we have $C^*\cap I = \emptyset$.
\end{proof}
A cycle $C$ is said to be {\em minimal} if it is minimal
amongst all cycles with respect to $\subseteq$ (that is,
whenever $C'\subseteq C$ and $C'$ is a cycle, then $C'=C$).
Sometimes, minimal cycles are called {\em elementary}, 
for instance in \cite{jo}.
\begin{remark}
It is easy to see that Theorem \ref{chara1} can be made a bit
simpler: in order to verify that a certain subset $I\subseteq V$
is filling, it suffices to check that $A(G)\subseteq I$ and
that $I$ intersects every {\em minimal} cycle. 
\end{remark}
\begin{example}\label{ex:chara1}
	In the introductory example of a weight-calculator, we have considered the three input fields ``Sex'', ``Age'' and ``Height'' along with two replacement functions. As shown on the left hand side of Figure~\ref{fig:intro}, the fields are represented by the three vertices in the graph $G$. The directed edge from ``Height'' to ``Age'' stands for the replacement function $f$ for age based on height. Finally, the replacement function $g$ for height depending on age and sex is represented by the two other edges pointing to ``Height''. By our definitions, we can for instance make the following statements:
	\begin{enumerate}
		\item The vertex set $\{\text{Sex}, \text{Age}\}$ determines $\{\text{Height}\}$.
		\item The vertex $\{\text{Height}\}$ determines $\{\text{Age}\}$.
		\item The set $A(G)$ equals $\{\text{Sex}\}$ (no edge points to it).
		\item By Theorem~\ref{chara1}, the vertex set $\{\text{Sex}\}$ is not filling, i.\,e.\ by entering only sex, the other fields cannot be autofilled. The reason is that the set $\{\text{Sex}\}$ does not intersect the cycle formed by the two vertices $\{\text{Age}, \text{Height}\}$.
		\item There are exactly three filling subsets for $G$: $\{\text{Sex}, \text{Age}\}$, $\{\text{Sex}, \text{Height}\}$ and (trivially) also $\{\text{Sex}, \text{Age}, \text{Height}\}$ as they all contain $A(G)$ and at least one element of the only cycle.
	\end{enumerate}
\end{example}
For complex graphs, the application of Theorem~\ref{chara1} might need algorithmic support to identify $A(G)$ and particularly the cycles to check during runtime if partial input already fills the remaining fields. To do so, the following algorithms can be applied.
\begin{algorithm}[Identifying $A(G)$]\label{alg:chara1_1}
A convenient representation for a directed graph is the {\em adjacency matrix}:
The vertex are numbered $1,\ldots, n$ and we assign a binary $n\times n$-matrix 
$$M_G\in \mathbb{Z}_2^{n\times n}$$ to $G$ by setting $M_G[i,j]= 1$ if
and only if $(i,j)\in E$, and $M_{G}[i,j] = 0$ otherwise.

Now $A(G)$ is easily identified: $i\in A(G)$ if and only if $M_G[\cdot, i]$ (that is the $i$th
column vector) is the constant $0$ vector. 
\end{algorithm}
\begin{algorithm}[Identifying cylces]
	 Different algorithms exist for finding cycles in a graph, see e.\,g.~\cite{jo} for a solution and further references. 
\end{algorithm}

The next result link graph filling to the concept of directed acyclic graph (DAG).
\begin{theorem}[Connection to DAGs]\label{th:dag}
	Let $G=(V,E)$ be a self-loopless directed graph, and let $I\subseteq V$. Furthermore denote by $G' = (V, E')$ with $E' = E \setminus \{(x,i) : x\in V \text{ and } i\in I\}$ the subgraph without edges pointing to any $v \in I$. Then the following are equivalent: 
	\begin{enumerate}
		\item $I$ is filling in $G=(V,E)$;
		\item $I$ is filling in $G'=(V,E')$ and $G'$ is a DAG.
	\end{enumerate}
\end{theorem}
\begin{proof}
	
$(1) \implies (2)$. First, we prove that $I$ is filling in $G=(V,E')$.
Let $I'_n$ be the determination closures
of $I$ in the graph $G'$. With $I_n$ we denote the determination closures
of $I$ in $G$. For $v\in V$ we denote by $\In'(v)$ the set
of incoming vertices in the graph $G'$, and by $\In(v)$ the
set of incoming vertices in $G$. 
As $I$ is filling in $G$ by assumption, we have
$I_N = V$ for some $N\in\mathbb{N}$, so it suffices to show that 
$$I'_n \supseteq I_n \textrm{ for all }n\in\mathbb{N}.$$ We proceed
by induction. Clearly $I'_0 = I_0 = I$. Suppose we have
$I'_k \supseteq I_k$ for some $k$ and show that $I'_{k+1}
\supseteq I_{k+1}$. Let $v\in I_{k+1}$. If $v\in I_k$ we
get $v\in I'_{k+1}$ automatically since $v\in I_k
\subseteq I'_k \subseteq I'_{k+1}$. If $v\notin I_k$
then we have trivially $v\notin I$. Moreover,
the fact that $I_k$ determines $v$ in $G$ implies
in particular that
$\In(v) \neq \emptyset$, so we pick some
incoming vertex $j\in \In(v)$. Since $v\notin I$ we have
$(j, v)\in E'$ by the definition of $E'$, so 
$$(A) \hspace*{1cm} \In'(v)\neq \emptyset 
\textrm{ in } G'.$$ Since $I_k$ determines
$v$ in the graph $G$, we get
$\In(v) \subseteq I_k$, and by
definition of $G'$ we get $\In'(v) \subseteq \In(v)$.
Combining this gives $\In'(v) \subseteq I_k \subseteq I'_k$
by induction assumption, and together with $(A)$
this implies that $I'_k$ determines $v$ in the graph $G'$, 
so $v\in I'_{k+1}$, which finishes the proof of (1).

Next, we show that $G'$ is a DAG. 
Assume that $G'$ contains a cycle $C$. By Theorem 
\ref{chara1} we know that $I\cap C\neq\emptyset$, so
pick $v^*\in I\cap C$. Because $C$ is a cycle there
is a bijection $p:\{0,\ldots, n\}\to C$ such that 
$(p(k),p(k+1))\in E'$ for $k\in\{0,\ldots,n-1\}$ and
$(p(n), p(0))\in E'$. We can pick $p$ such that 
$p(0) = v^*$. But by definition of $E'$ we have
$(p(n), v^*) \in E \setminus E'$, contradiction.

$(2) \implies (1)$. This is easily verified by 
noticing that $\Fill(I)$ in the graph $G'=(V,E')$ equals $\Fill(I)$ in 
$G=(V,E)$ 
(we don't even need the assumption that $G'=(V,E)$ is a DAG).
\end{proof}

Since a DAG $G' = (V, E')$ does not contain cycles, Theorem~\ref{chara1} ensures that $I \subseteq V$ is filling if and only if $A(G') \subseteq I$.

This provides a second possibility to verify if a subset $I \subseteq V$ of vertices of a directed graph $G = (V, E)$ is filling or not: Delete all edges pointing to vertices in $I$ and check if the resulting subgraph $G'$ is a DAG with $A(G') \subseteq I$. If yes, $I$ is filling.
 
\begin{example}\label{ex:dag}
	In Example~\ref{ex:chara1} we have utilized Theorem~\ref{chara1} to show that, amongst others, the set $I := \{\text{Sex}, \text{Age}\}$ is filling. Alternatively, we consider the subgraph $G'$ without edges pointing to $I$ (see left side of Figure~\ref{fig:dag}). Obviously, $G'$ is a DAG with $I=A(G')$ and thus, $I$ is filling in the original graph $G$. The right side of Figure~\ref{fig:dag} shows that $I' = \{\text{height}\}$ alone is not filling: Although the subgraph $G''$ is a DAG, the subset $A(G'') = \{\text{Sex}, \text{Height}\} \not \subseteq I'$.
	\begin{figure}[h]
		\hspace{2cm}
		\centering
		\resizebox{0.25\textwidth}{!}{%
			\begin{tikzpicture}[->,>=stealth',shorten >=1pt,auto,node distance=3cm, thick,main node/.style={circle,draw}]	
			\node[main node] (1) [shading angle = 45]{Sex};
			\node[main node] (2) [shading angle = 45, below left of=1] {Age};
			\node[main node] (3) [below right of=2] {Height};
			
			\path[every node/.style={font=\sffamily\small}]
			(1) edge node [left] {} (3)
			(2) edge [bend right] node {} (3);
			\end{tikzpicture}}
		\hfill
		\resizebox{0.25\textwidth}{!}{%
			\begin{tikzpicture}[->,>=stealth',shorten >=1pt,auto,node distance=3cm, thick,main node/.style={circle,draw}]	
			\node[main node] (1) {Sex};
			\node[main node] (2) [below left of=1] {Age};
			\node[main node] (3) [shading angle = 45, below right of=2] {Height};
			
			\path[every node/.style={font=\sffamily\small}]
			(3) edge [bend right] node {} (2);
			\end{tikzpicture}}
		\hspace{2cm}
		\caption{On the left: Subgraph $G'$ without edges pointing to candidate set $I := \{\text{Sex}, \text{Age}\}$ (highlighted). Right: Subgraph $G''$ without edges pointing to $I' = \{\text{Height}\}$. Both subgraphs are DAGs. Since $A(G') \subseteq I$, $I$ is filling. But $A(G'') \not \subseteq I$, thus $I'$ is not filling.}
	\end{figure}
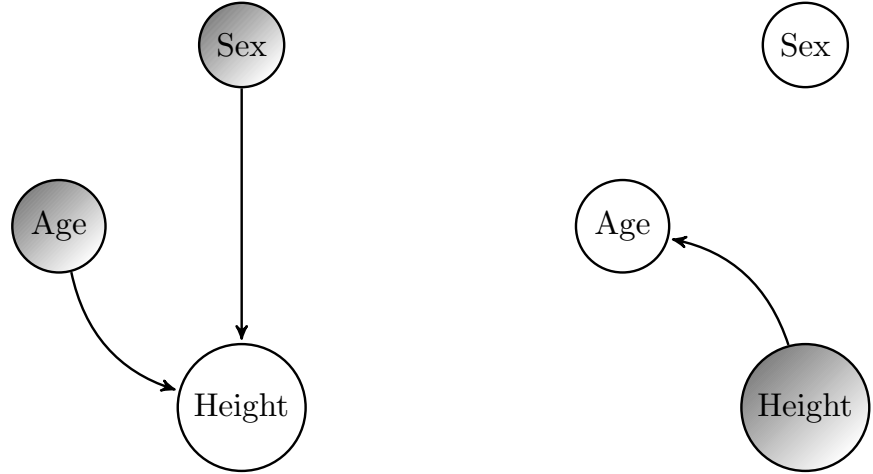\label{fig:dag}
\end{example}
While in our simple examples it is easy to verify if a (sub-)graph is a DAG, in more involved settings the following algorithm could be used systematically.

A simple algorithm to check if $G$ is a DAG uses the following algorithm.
\begin{algorithm}[Directed path]\label{alg:directed.path}
Whenever one deals with directed paths in graphs, 
Dijkstra's Algorithm as described by himself in \cite{Dij} 
is one of the most useful tools: it finds the shortest
path (if there is any) between vertices.
So let $G$ be a directed graph on $n$ vertices
and  let $M_G$ be its adjacency matrix.
There is a directed path from vertex $i$ to vertex $j$ if and only
if $M_G^k[i,j] > 0$ for some $k\in\{1,\ldots,n-1\}$. 	
\end{algorithm}
\begin{algorithm}[Checking whether $G$ is a DAG]\label{alg:dag}
Since in a DAG with $n$ vertices,  we have
no path of length $n$, Algorithm~\ref{alg:directed.path} gives the following criterion: 
$G$ is a DAG if and only if the sum of all traces of $M_G^1, M_G^2, \dots, M_G^{n-1}$ is 0. Other algorithms exist to identify if a directed graph is acyclic, see e.\,g.\ \cite{Ta76}.
\end{algorithm}
\begin{definition}
	We say that a filling subset $I\subseteq V$
	is {\em minimal}
	(with respect to $\subseteq$) if no
	proper subset of $I$ is filling.
\end{definition}
The following example shows that two different
minimal filling subsets do not necessarily
have the same number of elements:
\begin{example}\label{example_minind}
	Let $V=\{1,2,3\}$ and $E = \{(1,2),(2,1)\} 
	\cup \{(2,3),(3,2)\}$. 
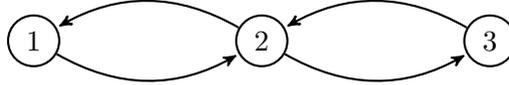
	\begin{center}
		\begin{tikzpicture}[->,>=stealth',shorten >=1pt,auto,node distance=3cm,
		thick,main node/.style={circle,draw}]
		
		\node[main node] (1) {1};
		\node[main node] (2) [right of = 1] {2};
		\node[main node] (3) [right of=2] {3};
		
		\path[every node/.style={font=\sffamily\small}]
		(1) edge [bend right] node  {} (2)
		(2) edge [bend right] node  {} (1)
		(2) edge [bend right] node  {} (3)
		(3) edge [bend right] node  {} (2);
		\end{tikzpicture}
	\end{center}
	Then $I = \{1,3\}$ is
	minimal filling, and also $J= \{2\}$, but
	$I$ and $J$ differ in cardinality.
	Note that, even if we could visit each vertex of the graph by starting at $K = \{1\}$ and following the directed edges, $K$ is not filling: The reason is that $\{2\}$ is not completely determined by $K$ alone (there is also edge pointing from $\{3\}$ to $\{2\}$ and $\{3\}$ is only determined through $\{2\}$). Put differently, $K_1 = K \cup \Fill(\{1\}) = \{1\}$, thus $K_1 = K_2 = K_3 = \dots = \{1\}$. Therefore $K = \{1\}$ is not filling.
\end{example}
Can we choose  a minimal filling subsets such that it intersects every
cycle (or every minimal cycle) at exactly 1 vertex? Unfortunately not:
\begin{example}
Let $V = \{0,1,2\}$, let $$E = (V\times V)\setminus 
\big\{(k,k): k\in\{0,1,2\}\big\},$$
and let $G=(V,E)$.  Note that $A(G) = \emptyset$.
If we take $K = \{k\}$ for some $k \in \{0,1,2\}$, it is easily verified
that $\Fill(K)=\emptyset$. So $K_n = K$ for all $n\in\mathbb{N}$, 
and therefore $K$ is not filling. So if $I\subseteq V$ is minimal filling, it
contains at least $2$ vertices of $V=\{1,2,3\}$, say $\{k_1,k_2\}\subseteq I$
for $k_1\neq k_2\in\{1,2,3\}$. But then $C=\{k_1,k_2\}$ is a minimal cycle
by definition of $E$, and $|I\cap C| > 1$.
\end{example}
By Theorem~\ref{chara1} or Theorem~\ref{th:dag} it is not hard to verify whether a given set of vertices is filling or not. However, in general there will be no simple algorithm to identify the smallest possible filling subset of $V$. The following greedy algorithm will, however, usually provide a good approximation.
\begin{algorithm}
	\mbox{}
	\begin{enumerate}
		\item Identify the set $I = A(G)$ of all vertices without incoming edges.
		\item\label{la:step2} Choose all vertices $W = \{w_1, \dots, w_m\}$ in all minimal cycles that do not intersect $I$. Denote by $n_i$ the number of cycles intersected by $w_i$, $1 \leq i \leq m$.
		\item If $W \ne \emptyset$, pick any $w_i \in W$ with maximal $n_i$ and set $I := I \cup \{w_i\}$
		\item Go to Step~\ref{la:step2} as long as $W$ contains at least two elements. Otherwise stop with $I$ as solution.
	\end{enumerate}
\end{algorithm}
The algorithm will terminate after maximal $n-1$ iterations with $n$ being the number of vertices. Due to its greedy nature, it might miss the optimal solution in certain hypothetical cases.
\begin{example}\label{ex:algorithm} To illustrate the algorithm and also the power of Theorems~\ref{chara1} and \ref{th:dag}, take the graph $G$ on the right hand side of Figure~\ref{fig:intro}. There are no vertices without incoming edge, thus $A(G)$ is empty. There are four minimal cycles that contain the following vertices. 
	\begin{enumerate}
		\item $\{\text{Age, Pregnant}\}$
		\item $\{\text{Sex, Pregnant}\}$
		\item $\{\text{Height, Age, Pregnant}\}$
		\item $\{\text{Sex, Height, Age, Pregnant}\}$
	\end{enumerate}
	Consequently, after the first iteration of the algorithm up to Step 3, the candidate filling set $I$ consists of the vertex ``Pregnant'' (hits highest number of cycles and $A(G)$ is empty). Furthermore, $W$ consists all vertices, so we start with a second iteration but now $W$ is left empty and the algorithm stops. Thus, even if only the field ``Pregnant'' is entered, the remaining fields can be autofilled without further input. Of course there are also filling subsets without containing ``Pregnant'', e.\,g.\ $I'=\{\text{Age, Sex}\}$ (apply Theorem~\ref{chara1}). Figure~\ref{fig:algorithm} presents the two subgraphs without edges pointing to $I = \{\text{Pregnant}\}$ (left picture) resp.\ to $I'$ (right picture), illustrating the idea of Theorem~\ref{th:dag}.
\end{example}
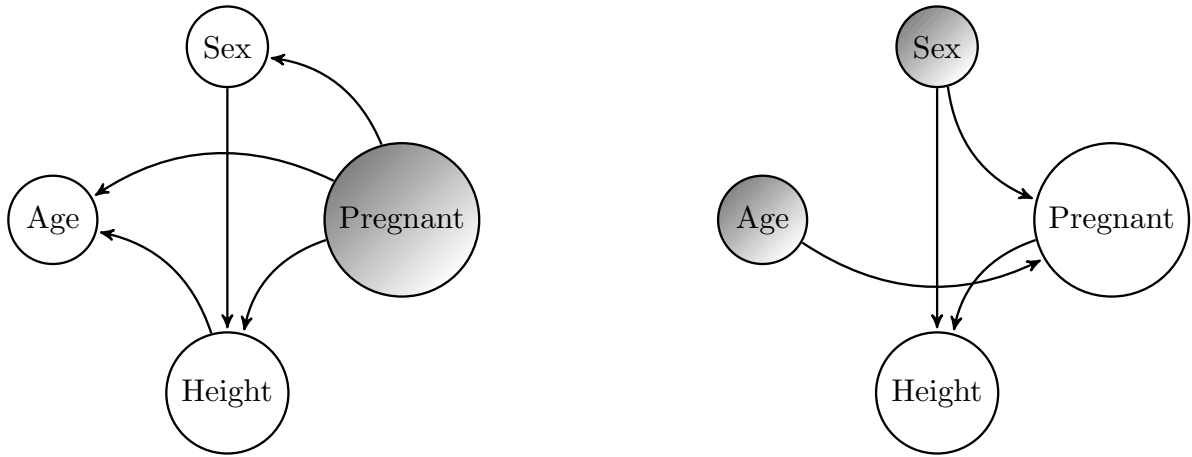
\begin{figure}[h]
	\centering
	\resizebox{0.4\textwidth}{!}{%
		\begin{tikzpicture}[->,>=stealth',shorten >=1pt,auto,node distance=3cm, thick,main node/.style={circle,draw}]	
		\node[main node] (1) {Sex};
		\node[main node] (2) [below left of=1] {Age};
		\node[main node] (3) [below right of=2] {Height};
		\node[main node] (4) [shading angle = 45, below right of=1] {Pregnant};
		
		\path[every node/.style={font=\sffamily\small}]
		(1) edge node [left] {} (3)
		(3) edge [bend right] node {} (2)
		(4) edge [bend right] node {} (1)
		(4) edge [bend right] node {} (2)
		(4) edge [bend right] node {} (3);
		\end{tikzpicture}}
	\hfill
	\resizebox{0.4\textwidth}{!}{%
		\begin{tikzpicture}[->,>=stealth',shorten >=1pt,auto,node distance=3cm, thick,main node/.style={circle,draw}]	
		\node[main node] (1) [shading angle = 45]{Sex};
		\node[main node] (2) [shading angle = 45, below left of=1] {Age};
		\node[main node] (3) [below right of=2] {Height};
		\node[main node] (4) [below right of=1] {Pregnant};
		
		\path[every node/.style={font=\sffamily\small}]
		(1) edge node [left] {} (3)
		(4) edge [bend right] node {} (3)
		(1) edge [bend right] node {} (4)
		(2) edge [bend right] node {} (4);
		\end{tikzpicture}}
	\caption{Subgraphs resulting from removing the incoming edges to $I = \{\text{Pregnant}\}$ resp.\ $I'=\{\text{Age, Sex}\}$ (both highlighted) in the situation of Example~\ref{ex:algorithm}. Since $I$ resp.\ $I'$ contain the vertices without incoming edges and since the subgraphs are acyclic, both $I$ and $I'$ are filling in the original graph without removed edges}.
\end{figure}\label{fig:algorithm}
\section{Filling with partial determination}\label{se:theory:partial}
Given $I\subseteq V$ and $v\in V$ we say that $I$ {\em
determines} the vertex $v$ {\em 
with partial input} if $\In(v) \cap I \neq \emptyset$, that is if {\em any} argument in a replacement function is provided.
For short, we say in that case that $I$ {\em p-determines}
$v$.

We denote the set of vertices p-determined by $I$ by
$\pFill(I)$.
Inductively we define the ``p-determination closure'' 
of $I\subseteq V$ in the following way:
\begin{enumerate}
\item $I^{(p)}_{0} := I$;
\item for $n\geq 0$ set $I^{(p)}_{n+1} :=I^{(p)}_n \cup \pFill(I_{n})$.
\end{enumerate}
We say that $I\subseteq V$ {\bf p-fills} $V$ if there is $n\in
\mathbb{N}$ such that $I^{(p)}_{n} = V$.
\begin{example}
	In Example~\ref{example_minind}, $I = \{1\}$ is p-filling (but not filling) because $I_1 = \{1, 2\}$ and $I_2 = \{1, 2, 3\}$.
\end{example}
There is a first, trivial characterization of p-filling
subsets of a graph $G=(V,E)$:
\begin{proposition}
Let $G=(V,E)$ be a self-loopless directed graph, and let $I\subseteq V$.
Then the following are equivalent:
\begin{enumerate}\label{way_prop}
\item $I$ is p-filling;
\item $A(G)\subseteq I$, and for every vertex $x\in V\setminus I$ 
there is a directed path from some $j\in I$ to $x$.
\end{enumerate}
\end{proposition}

However, we can do much better than that. Next, we specify in
a mathematical way what it means to ``collapse'' points $x,y\in V$
when there are directed paths between theses points in either direction.

\begin{definition}
If $G=(V,E)$ is a directed graph and $x,y\in V$ we say
that $x,y$ are {\em strongly connected}, in symbols
$x\simeq y$ if there
exists a directed path from $x$ to $y$ and vice versa.
\end{definition}
Again, it is straightforward to verify that $\simeq$ is
an equivalence relation on $V$. The set of elements
equivalent to $x\in V$ is denoted by $[x]_{\simeq}$,
and we call it the {\em strongly connected component (scc)}
containing $x$. The set of scc's on $V$ with
respect to $\simeq$ is denoted by $V/_{\simeq}$.
Note that by construction, every $v\in V$ lies in a unique
scc, the scc's are mutually disjoint, and all scc's 
are non-empty.

We put a directed graph structure on $V/_{\simeq}$
and set $G/_{\simeq} = (V/_{\simeq}, E/_{\simeq})$ where
$$E/_{\simeq} = \{(C,D)\in V/_{\simeq}\times V/_{\simeq}:
C\neq D \textrm{ and there are } c\in C, d\in D 
\textrm{ such that }
(c,d)\in E\}.$$
The following is an elementary observation:
\begin{lemma}\label{daglemma}
If $G$ is a directed graph, $G/_{\simeq}$ is a DAG.
\end{lemma}
Before we show how strongly connected components come into play
for finding p-filling sets, we need some basic observations.
\begin{lemma}\label{folklore1}
\hspace*{1cm} 
\begin{enumerate} 
\item If $G=(V,E)$ is a DAG, then for all $v\in V$ there is
$x\in A(G)$ and a directed path from $x$ to $v$.
\item Let $G=(V,E)$ be any graph. If $C,D\in G/_{\simeq}$
and there is a direct path in $G/_{\simeq}$ from $C$ to $D$
then for all $c\in C, d\in D$ there is a directed path
in $G$ from $c$ to $d$.
\end{enumerate}
\end{lemma}
Combining these observations lead us to the following:
\begin{proposition}\label{key21}
Let $G=(V,E)$ be any graph, and let $v\in V$. Then 
there is a strongly connected component $C_0 \in 
A(G/_{\simeq})$ such that for every $c_0\in C_0$ there 
is a directed path in $G$ from $c_0$ to $v$.
\end{proposition}
\begin{proof}
Let $[v]_{\simeq}$ be the strongly connected component
containing $v$. Lemma \ref{daglemma} says that $G/_{\simeq}$
is a DAG. By Lemma \ref{folklore1} (1) 
there is $C_0\in A(G/_{\simeq})$
and a directed path in $G/_{\simeq}$ from $C_0$
to $[v]_{\simeq}$. Finally, Lemma \ref{folklore1} (2) 
implies that there is a directed path in $G$
from any $c_0\in C_0$ to $v$.
\end{proof}
\begin{theorem}\label{chara2}
Let $G=(V,E)$ be a directed graph, $I\subseteq V$,
and consider the 
graph $G/_{\simeq}$.
Then the following
statements are equivalent:
\begin{enumerate}
\item $I$ is p-filling;
\item $I$ intersects every strongly connected component
$C\in A(G/_{\simeq})$.
\end{enumerate}
\end{theorem}

\begin{proof}

\hspace*{2cm}

(1) $\implies$ (2). Suppose $C^* \in A(G/_{\simeq})$ 
and let $I\subseteq V\setminus C^*$. We show that
$I$ is not p-filling for $G$ by establishing that 
$I^{(p)}_n \cap C^* = \emptyset$ for all $n\in \mathbb{N}$.
The statement is true for $I^{(p)}_0 = I$. 
Assume that $I^{(p)}_k\cap C^* = \emptyset$. The fact
that $C^*\in A(G/_{\simeq})$ means by definition of $A(\cdot)$
and by definition of $E_{\simeq}$ that there is no
edge in $E$ coming into $C^*$ from the outside -- or,
more precisely, 
for all $x\in V\setminus C^*$ and $c\in C^*$ we have
$(x,c)\notin E$. Therefore, for all $c\in C^*$ we have
$c\notin I^{(p)}_{k+1}$, that is $C^*\cap I^{(p)}_{k+1}
=\emptyset$. This inductive argument proves that
$I^{(p)}_n\cap C^* = \emptyset$ for all $n\in \mathbb{N}$,
so $I^{(p)}_n \neq V$ for all $n$, and therefore $I$
is not p-filling.

(2) $\implies$ (1). Fix $v\in V$.  
By Proposition \ref{way_prop} we need 
to establish that if $I$ intersects any strongly connected
component (= element of $A(G/_{\simeq})$), then there
is a directed path from some $i\in I$ to $v$. The 
proposition then implies that $I$ is filling.

Use Proposition \ref{key21} to find $C_0\in A(G/_{\simeq})$
such that for every $c_0\in C_0$ there is a directed
path in $G$ from $c_0$ to $v$. Since $I$ intersects
$C_0$ by assumption, pick $j_0 \in I\cap C_0$. So
there is a directed path in $G$ from $j_0\in I$ 
to $v$. So by Proposition \ref{way_prop}, $I$ is filling
because $v$ was chosen arbitrarily.
\end{proof}
\begin{algorithm}
In \cite{Ta}, Tarjan introduced an algorithm
that identifies for any graph $G=(V,E)$ its
strongly connected components in time $O(|V|+|E|)$.
\end{algorithm}
Note that this theorem implies that minimal p-filling
subsets intersect every member of $A(G/_{\simeq})$ at
exactly one point. So this implies:
\begin{corollary}
All minimal (with respect to $\subseteq$) p-filling
subsets have the same cardinality.
\end{corollary}
This is in sharp contrast to the situation in the 
previous section, where minimal filling subsets can
have different cardinalities (see Example \ref{example_minind}).

Moreover, Theorem \ref{chara2} gives an efficient algorithm to find
minimally p-filling sets: Identify the strongly connected
components that don't have an incoming edge
(that is, $A(G/_{\simeq})$), and pick a vertex from each.
\begin{example}
	In the introductory example of a weight-calculator with the vertices 
	$V = \{\text{Age, Height, Sex}\}$, we could modify the replacement function
	\begin{align*}
		g(x, s) := 
		\begin{cases}
			162 + 16s & \text{if } x > 16, \\
			\lfloor (x - 1)/16 \cdot 130 + 30.5\rfloor & \text{if } x \le 16
		\end{cases}
   \end{align*} 
   for height $z$ based on (non-missing) age $x$ and sex $s$ by a function that allows one of the two arguments to be missing, for instance by
   \begin{align*}
   		g'(x, s) := 
   		\begin{cases}
   			162 + 16s & \text{if } (x > 16 \text{ or $x$ missing$)$ and $s$ non-missing}, \\
   			170 & \text{if } x > 16 \text { and $s$ missing}, \\
   			\lfloor (x - 1)/16 \cdot 130 + 30.5\rfloor & \text{if } x \le 16.
   		\end{cases}
   	\end{align*} 
   	Then, in our graph-theoretic autofill framework, the vertex ``Height'' would partially be determined by ``Sex'' and ``Age'' and ``Age'' (partially) determined by ``Height''. The graph on the left side of Figure~\ref{fig:intro} would be equivalent (under $\simeq$) to the DAG $G/_{\simeq}$ with the two strongly connected components $\{\text{Sex}\}$ and $\{\text{Age, Height}\}$ as vertices. By Theorem~\ref{chara2}, any subset of $V$ containing $A(G/_{\simeq}) = \{\text{Sex}\}$ would be filling. 
\end{example}
Generally, to use partial determination requires more effort to properly define the replacement functions compared to complete determination as these functions also need to treat missing input in a reasonable way. However, usually a smaller subset of input values is required to fill the remaining values.

\section{Conclusion}
Loopless directed graphs turn out to be the ideal framework for representing different kinds of ``inference'' or ``implication''. In this article, we used directed graphs in the context of missing values in vectors. In web forms used today, the set of mandatory fields is fixed. Our approach offers something new: flexible and smart filling of missing values.


{\footnotesize
\begin{thebibliography}{99}
\bibitem{Di} R.~Diestel. {\bf Graph Theory}, Springer Verlag, 2010.
\bibitem{jo} D.~Johnson. {\em Finding all the elementary circuits of a directed graph}, SIAM Journal on Computing 4(1): 77--84, 1975.
\bibitem{Dij} E.~Dijkstra. {\em A note on two problems in connexion with graphs},
Numerische Mathematik 1: 269--271, 1959.
\bibitem{Ta76} R.~E.~Tarajan. {\em Edge-disjoint spanning trees and depth-first search}, Acta Informatica 6(2): 171--18, 1976.
\bibitem{SP} M.~Saar-Tsechansky and F.~Provost. {\em Handling Missing Values when Applying Classification Models}, Journal of Machine Learning Research 8: 1625--1657, 2007.
\bibitem{SG} D.~Schuurmans and R.~Greiner. {\em Learning to classify incomplete examples}, Computational Learning Theory and Natural Learning Systems IV: 
Making Learning Systems Practical: 87--105, MIT Press, Cambridge MA, 1997.
\bibitem{Ta} R.~E.~Tarjan. {\em Depth-first search and linear algorithms for graphs}, SIAM Journal on Computing, 1(2): 146--160, 1972.
\end{thebibliography}
}
\end{document}